\documentclass[runningheads]{llncs}

\usepackage[T1]{fontenc}
\usepackage{graphicx}
\usepackage{amsmath,amssymb}
\usepackage{subcaption,booktabs,multirow}

\begin{document}
\title{Duration-Aware Dissimilarity in Temporal Networks}
%
%
\author{Giuseppe F. Italiano\inst{1} 
\and
Anna Mpanti\inst{2}
}
\authorrunning{Giuseppe F. Italiano et al.}
%
\institute{LUISS University, Italy
\email{gitaliano@luiss.it} 
 \and University of Ioannina, Greece
 \email{a.mpanti@uoi.gr}
}

\maketitle              

\begin{abstract}
Temporal-network comparison becomes substantially more challenging when interactions are associated with explicit traversal times, since both event timing and duration determine temporal reachability and journey performance. We introduce a duration-aware dissimilarity framework for temporal networks whose edges carry departure times and traversal durations. The proposed approach compares networks through source-specific temporal-distance profiles using the $1$-Wasserstein distance, thereby retaining both temporal information and source correspondence.

We evaluate the framework on real public-transport timetable networks under controlled schedule delays, duration-aware randomized reference models, and naturally occurring temporal variations. We further benchmark the proposed measure against existing path-based temporal dissimilarity approaches under the same perturbation scenarios. The results show that the proposed representation responds consistently to both the extent and the magnitude of temporal disruption, while retaining distinctions that may be compressed by contact- or hop-based temporal summaries. Overall, the findings support duration-aware temporal distances as a useful basis for comparing temporal systems in which traversal and waiting times directly affect network performance.


\end{abstract}

\section{Introduction}

Temporal networks model systems in which interactions are available at
specific points in time. In contrast to static graphs, the order of interactions determines whether information, individuals, or resources can travel between two vertices. Consequently, temporal paths and temporal reachability may differ substantially even between networks having the same underlying static graph \cite{HoSa,KeKlKu,WuChHuKeLuXu}.

In many applications, however, the time at which an interaction becomes
available is only part of the relevant temporal information. Traversing an interaction may itself require a non-negligible time. A public transportation connection, for example, is naturally described by its departure stop, arrival stop, departure time, and travel duration. The duration determines not only the total travel time, but also whether subsequent connections can be reached. Similar considerations arise in communication, routing, and logistics networks~\cite{Ba,OrRo}.

This observation exposes two difficulties in comparing temporal networks with traversal times. First, two networks may have identical interaction endpoints and departure times, yet different traversal durations; a representation based only on instantaneous contacts would regard them as identical, although their feasible temporal paths, waiting times, and fastest travel times may differ. Second, even when the temporal contact structure changes, the operational effect of the change depends on the temporal paths induced by the modified schedule: delaying a connection may simply shift one event, or it may invalidate subsequent transfers and substantially increase end-to-end journey times. This motivates comparing duration temporal networks through the temporal-distance profiles they induce, rather than only through the organization of their instantaneous contacts. Since fastest travel times lie on an ordered numerical domain, we use the $1$-Wasserstein distance to compare these profiles while retaining the magnitude of temporal displacement.


\subsection{Related work}

Temporal networks and time-respecting paths have been studied extensively from both network-science and algorithmic perspectives
\cite{HoSa,KeKlKu}.
Unlike static graphs, temporal graphs admit several non-equivalent notions of an optimal path, including earliest-arrival, latest-departure, fastest, and shortest paths~\cite{WuChHuKeLuXu}. 

Network comparison has been studied extensively in the static setting, with methods based on structural distance distributions
~\cite{ScCaDiOaNaRa} and multiscale graph representations~\cite{BaBo}. Such approaches, however, do not capture the temporal ordering or duration of individual interactions. Temporal-network comparison has therefore received separate attention. Froese et al.~\cite{FrJaNiRe} combined vertex signatures with dynamic time warping to compare sequences of graph snapshots. More recently, Zhan et al.~\cite{ZhLiWaWaHoZh} proposed a temporal dissimilarity based on first-arrival-distance distributions, temporal node dispersion, and the Jensen--Shannon divergence. Their approach preserves information lost under static aggregation and was shown to be useful for temporal-network classification and the comparison of spreading processes. Among existing temporal-network dissimilarities, the approach of Zhan et al.~\cite{ZhLiWaWaHoZh} is the closest baseline to our setting, since it also characterizes a temporal network through distributions derived from temporal paths. Its first-arrival distance, however, is defined on instantaneous contacts $(u,v,t)$ and records the number of contacts on an earliest-arrival path. In contrast, our setting assigns an explicit traversal duration $\delta$ to every interaction and summarizes paths through their elapsed travel time, including both traversal and waiting. The distinction is therefore not merely the choice of a distributional distance: the two approaches retain different information about the temporal paths induced by the network. We compare them directly in Section~\ref{sec:experimental-results}.


Optimal-transport distances provide a natural way to compare distributions supported on an ordered domain, since they account for
the displacement of probability mass rather than only for differences in probability values; the $1$-Wasserstein distance is a standard instance of this framework~\cite{GaMa}. Randomized reference models have also been widely used in temporal-network analysis to isolate the effect of specific temporal features by selectively randomizing them while preserving other network characteristics~\cite{GaGeKaKiTaVaVe}.



\subsection{Our contribution} 
\label{sec:contribution} 

We study the comparison of temporal networks in which interactions have explicit, non-zero traversal durations. Our framework represents each source vertex through a reachability-aware temporal-distance profile and compares corresponding sources while preserving their identity. We focus primarily on fastest temporal distances, so that the representation reflects traversal times, intermediate waiting, and changes in temporal reachability. 

We evaluate the framework on four public-transport timetable networks using controlled delays, duration-aware randomized reference models, and day-to-day schedule variations. Unlike a purely internal sensitivity analysis, we also benchmark the proposed $\operatorname{TD}_{\mathrm{fast}}$ measure against the FAD/JS temporal dissimilarity of Zhan et al.~\cite{ZhLiWaWaHoZh} on exactly the same controlled-delay realizations. We additionally consider a duration-aware adaptation of the FAD/JS construction, which uses traversal durations to determine earliest-arrival paths but continues to summarize them by their number of connections. The experiments show that all approaches respond to broad temporal perturbations, while $\operatorname{TD}_{\mathrm{fast}}$ provides a more consistent distinction between perturbations of different temporal magnitude. Together with the randomized-reference-model analysis, these results show the value of retaining elapsed journey-time information when comparing duration temporal networks.

\section{Preliminaries}
\label{sec:preliminaries}

Classical event-based temporal networks represent an interaction by a triple $(u,v,t)$, indicating that vertices $u$ and $v$ interact at time $t$~\cite{HoSa,ZhLiWaWaHoZh}. 
In several applications, however, traversing an interaction requires a  non-zero amount of time, as illustrated by public-transportation timetables, where a connection has both a departure time and a travel duration.

Let $\mathbb{T}=[0,T]\subseteq\mathbb{R}_{\geq 0}$ be a bounded observation interval. A \emph{duration temporal network} is a pair $\mathcal{G}=(V,\mathcal{E})$, where $V$ is a finite set of vertices and $\mathcal{E}$ is a finite multiset of directed temporal edges. Each temporal edge is a quadruple $ e=(u,v,t,\delta)$, where $u,v\in V$ are its tail and head, respectively, $t\in\mathbb{T}$ is its departure time, and $\delta\in\mathbb{R}_{>0}$ is its traversal duration. The edge can be entered at time $t$ and reaches $v$ at time $t+\delta$. We assume that $t+\delta\leq T$.

For a temporal edge $e=(u,v,t,\delta)$, we write $\operatorname{dep}(e)=t$ and $\operatorname{arr}(e)=t+\delta$ for its departure and arrival times. Parallel temporal edges are allowed, since multiple interactions may connect the same ordered pair of vertices at different, or even identical, departure times. We use $n=|V|$ and $m=|\mathcal{E}|$. 
In a public-transportation network, vertices correspond to stops and a scheduled connection from stop $u$ to stop $v$, departing at time $t$ and arriving at time $t+\delta$, is represented naturally by the edge $(u,v,t,\delta)$~\cite{PyScWaZa}.

A \emph{temporal path} from $v_0$ to $v_k$ is a sequence $P=\langle e_1,e_2,\ldots,e_k\rangle$, where $ e_i=(v_{i-1},v_i,t_i,\delta_i)$, $
    i=1,\ldots,k,$ and consecutive temporal edges satisfy $ t_i+\delta_i\leq t_{i+1}, i=1,\ldots,k-1$. Thus, the traversal of an edge must be completed before the next edge can be entered. Waiting at intermediate vertices is allowed. We denote by $\mathcal{P}_{\mathcal{G}}(u,v)$ the set of temporal paths from $u$ to $v$ in $\mathcal{G}$. For $P=\langle e_1,\ldots,e_k\rangle$, we define
$    \operatorname{dep}(P) = t_1$, $
    \operatorname{arr}(P) = t_k+\delta_k$, $
    \operatorname{dur}(P) =
        \operatorname{arr}(P)-\operatorname{dep}(P)$, $
    |P|  = k$.
The duration $\operatorname{dur}(P)$ includes both traversal and
intermediate waiting times, whereas $|P|$ counts only temporal edges.

Temporal networks admit several non-equivalent notions of an optimal path \cite{WuChHuKeLuXu}. In this work, we consider two complementary distance notions. The \emph{fastest temporal distance} from $u$ to $v$ is
$
    d_{\mathrm{fast}}^{\mathcal{G}}(u,v)
    =
    \min_{P\in\mathcal{P}_{\mathcal{G}}(u,v)}
    \operatorname{dur}(P).
$
A fastest temporal path therefore minimizes the elapsed time between departure from $u$ and arrival at $v$. 
The \emph{minimum-hop temporal distance} from $u$ to $v$ is
$
    d_{\mathrm{hop}}^{\mathcal{G}}(u,v)
    =
    \min_{P\in\mathcal{P}_{\mathcal{G}}(u,v)}
    |P|.
$
It measures the minimum number of temporal interactions required to reach $v$ from $u$, subject to temporal feasibility. Waiting at an intermediate vertex does not increase the hop count. 
For either distance notion $d$, we adopt the conventions $d^{\mathcal{G}}(u,u)=0$ and $d^{\mathcal{G}}(u,v)=+\infty$ $\text{if}$ $\mathcal{P}_{\mathcal{G}}(u,v)=\emptyset$. Temporal distances are generally asymmetric, since temporal edges are directed and can be traversed only at their specified departure times.

Let $\mu=(\mu_1,\ldots,\mu_r)$ and $\nu=(\nu_1,\ldots,\nu_s)$ be two discrete probability distributions supported on $x_1,\ldots,x_r\in\mathbb{R}$ and
$y_1,\ldots,y_s\in\mathbb{R}$, respectively. A coupling between $\mu$ and $\nu$ is a non-negative matrix $\pi\in\mathbb{R}_{\geq 0}^{r\times s}$ satisfying $\sum_{j=1}^{s}\pi_{ij}=\mu_i$ for every $i$, $\sum_{i=1}^{r}\pi_{ij}=\nu_j$ for every $j$. We denote the set of all couplings between $\mu$ and $\nu$ by $\Pi(\mu,\nu)$. The $1$-Wasserstein distance between $\mu$ and $\nu$ is 
$
W_1(\mu,\nu)
=
\min_{\pi\in\Pi(\mu,\nu)}
\sum_{i=1}^{r}\sum_{j=1}^{s}
\pi_{ij}|x_i-y_j|$.

It can be interpreted as the minimum cost of transporting the probability mass of $\mu$ onto that of $\nu$, when moving one unit of mass from $x_i$ to $y_j$ incurs cost $|x_i-y_j|$ \cite{GaMa}. When the two distributions share a common ordered support
$x_1< x_2<\cdots<x_r$, let
$M_k=\sum_{i=1}^{k}\mu_i$, $
N_k=\sum_{i=1}^{k}\nu_i$
denote their cumulative masses. Then 
$W_1(\mu,\nu)
=
\sum_{k=1}^{r-1}
|M_k-N_k|(x_{k+1}-x_k)$.
Thus, $W_1$ accounts for the distance over which probability mass is displaced, making it suitable for distributions supported on ordered temporal distances. Moreover, it has the same physical unit as its support and is therefore measured in units of time for travel-time distributions.

\section{Duration-Aware Temporal Network Dissimilarity}
\label{sec:duration-aware-dissimilarity}

We now introduce a dissimilarity measure for duration temporal networks. Let $ \mathcal{G}_1=(V_1,\mathcal{E}_1)$ and $\mathcal{G}_2=(V_2,\mathcal{E}_2)$ be two duration temporal networks, and let $d$ be a temporal-distance notion. 
Our construction proceeds source-wise. For every source vertex, we summarize the temporal distances to all possible destinations by a discrete probability distribution. Finite temporal distances are assigned to distance bins, whereas temporally unreachable destinations are represented by an additional unreachable state. We then compare the profiles of corresponding sources using the 1-Wasserstein distance and aggregate the resulting source-level distances into a network-level dissimilarity.



\subsection{Temporal-Distance Profiles}

Fix a temporal distance $d$ and common bin boundaries $b_0<b_1<\cdots<b_r$, defining $B_j=[b_{j-1},b_j)$ for $j<r$ and $B_r=[b_{r-1},b_r]$. Let $U$ denote the common destination universe of the networks being compared. For each source $s$, we define a reachability-aware temporal-distance profile over the $r$ finite-distance bins and one additional state representing temporal unreachability.

Each destination $v\in U\setminus\{s\}$ contributes to the bin containing $d^G(s,v)$ when $d^G(s,v)<+\infty$, and to the unreachable state otherwise. The profile is normalized by $|U|-1$, rather than by the number of reachable destinations. Thus, the mass assigned to the unreachable state is exactly the fraction of destinations that cannot be temporally reached from $s$.

Each finite bin is represented by its midpoint $c_j=\frac{b_{j-1}+b_j}{2}$. The unreachable state is represented by an additional support point $c_\infty = c_r + \Delta_\infty$, where $\Delta_\infty>0$ is fixed throughout the experiment. In our
experiments, we set $\Delta_\infty=b_r-b_{r-1}$, i.e., the unreachable state is placed one final-bin width beyond the largest finite-distance representative.

This choice is fixed across all networks and perturbations: extending the support by one final-bin width makes transport from any finite bin to the unreachable state exactly one final-bin width more costly than transport to \(c_r\). Thus, becoming unreachable is penalized more strongly than any deterioration within the finite-distance range, without allowing reachability changes to dominate the comparison. Hence, all source profiles are defined on the common ordered support \(C=\{c_1,\ldots,c_r,c_\infty\}\).

All profiles compared within the same experiment use common bin boundaries. We use ten linearly spaced bins for the controlled-delay experiments and geometrically spaced bins for the randomized reference models. The binning choice is fixed across the networks compared in each experiment.

\subsection{Network-Level Dissimilarity}
\label{sec:network-level-dissimilarity}

We compare the reachability-aware profiles of corresponding source vertices in the two networks. Let $S$ be the common source set and let $H_s^{(1)}$ and $H_s^{(2)}$ denote the profiles of source $s$ in $G_1$ and $G_2$, respectively. Since unreachable destinations are represented explicitly, every source has a well-defined probability profile on the same support $\mathcal{C}$. For each $s \in S$, we define $w_s =
W_1\!\left(
H_s^{(1)}, H_s^{(2)}
\right)$.

Thus, probability mass must be transported not only when finite temporal distances change, but also when destinations become reachable or unreachable. In particular, a loss of temporal reachability transfers mass from the finite-distance part of the profile to the unreachable state and therefore contributes directly to the source-level dissimilarity.

\begin{definition}[Duration-aware temporal-network dissimilarity]
For a fixed temporal-distance $d$, we define $\operatorname{TD}_d(G_1,G_2)
=
\frac{1}{|S|L_{\mathcal C}}
\sum_{s\in S}
W_1\!\left(H_s^{(1)},H_s^{(2)}\right)$, where \(L_{\mathcal C}=c_\infty-c_1\) is the width of the common
reachability-aware support.
\end{definition}



\begin{proposition}
\label{prop:td-properties}
For a fixed source set $S$, fixed destination universe $U$, common support $\mathcal C$, and temporal-distance notion $d$, $0 \leq
\operatorname{TD}_d(\mathcal G_1,\mathcal G_2) \leq 1$. Moreover, $\operatorname{TD}_d$ is symmetric and satisfies
$\operatorname{TD}_d(\mathcal G,\mathcal G)=0$. On aligned collections of source profiles defined over the same
support $\mathcal C$, it also satisfies the triangle inequality.
\end{proposition}

\begin{proof}
For every source $s\in S$, both $H_s^{(1)}$ and $H_s^{(2)}$ are probability distributions on the common support $\mathcal C$.
Therefore, $0 \leq
W_1\!\left(H_s^{(1)},H_s^{(2)}\right)
\leq L_{\mathcal C}$, where $L_{\mathcal C}=c_\infty-c_1$ is the width of the common support. Averaging the source-level Wasserstein distances and normalizing by $L_{\mathcal C}$ gives the stated range. Symmetry and the triangle inequality follow source-wise from the corresponding properties of $W_1$ and are preserved under averaging. Finally, $\operatorname{TD}_d(\mathcal G,\mathcal G)=0$ follows because corresponding source profiles are identical.
\end{proof}

Different temporal networks may produce identical source profiles and therefore have zero dissimilarity. Hence, $\operatorname{TD}_d$ is in general a pseudometric on duration temporal networks represented using a fixed source set $S$, a fixed destination universe $U$, and a common support $\mathcal C$. The construction assumes aligned source identifiers and a common profile representation across the networks being compared.

\paragraph{Computational cost.}
Let $N_{\mathrm{TE}}$ and $M_{\mathrm{TE}}$ denote the numbers of vertices and edges, respectively, in the time-expanded graph. Since both travel and waiting edges advance forward in time, the time-expanded graph is acyclic. Fastest temporal distances from a source can therefore be computed by a single forward relaxation in temporal (topological) order, requiring $O(N_{\mathrm{TE}}+M_{\mathrm{TE}})$ time per source. In the time-expanded construction used here, each temporal edge induces at most two event vertices and one travel edge, while waiting edges connect consecutive events at each stop; hence $N_{\mathrm{TE}}=O(m)$ and $M_{\mathrm{TE}}=O(m)$. Thus, fastest temporal distances require $O(m)$ time per source and $O(|S|m)$ time over all sources. Given these distances, profile construction takes $O(|U|+r)$ time per source for the binning schemes used here, while the one-dimensional Wasserstein distance requires $O(r)$ time from cumulative masses. Therefore, after constructing the time-expanded graph, the overall source-wise computation and comparison requires $O(|S|(m+|U|+r))$ time, and the profiles can be processed one source at a time.


\section{Empirical Analysis}
\label{sec:empirical-analysis}

We evaluate the proposed dissimilarity on public-transport timetable networks using controlled delays, randomized reference models, and day-to-day variations. This section describes the experimental methodology; the corresponding results are presented in
Section~\ref{sec:experimental-results}.

\subsection{Experimental Methodology}
\label{sec:experimental-methodology}

We use public-transport timetable networks from the curated collection introduced by Kujala et al.~\cite{KuWeDaMlSa}, corresponding to Belfast, Grenoble, Luxembourg, and Turku. The collection is derived from GTFS timetable data and provides temporally filtered network event lists suitable for temporal-network analysis. Vertices represent public-transport stops, and each timetable record corresponds to a directed temporal edge, with traversal duration given by the difference between arrival and departure times.
Controlled perturbations and randomized reference models use the one-day temporal-network extracts, whereas the day-to-day comparison uses the corresponding week-long data partitioned by the weekday of each connection’s departure. Summary statistics of the four transportation networks are reported in Table~\ref{tab:datasets}.

For every network and source stop, we compute temporal distances to all destinations in the common stop universe and construct the corresponding reachability-aware profile described in Section~\ref{sec:duration-aware-dissimilarity}, with temporally unreachable destinations assigned to the additional unreachable state.

For fastest temporal distances, the weight of a travel edge in the time-expanded graph is its traversal duration, whereas the weight of a waiting edge is the corresponding waiting time. The profiles associated with the two networks being compared are aligned using their source-stop identifiers. If both networks arise from a perturbation of the same timetable, their vertex sets are normally identical. When comparing different days, we use the union of their stop sets as the common vertex universe. A stop that is absent from one of the two networks is treated as an isolated vertex in that network, so that destinations that cannot be reached from it, or destinations absent from that network, contribute to the unreachable state rather than to an empty profile.


We represent fastest temporal distances using $r=10$ bins. Linear bins
are used for controlled delays, with the same boundaries across all delay
ranges, perturbation levels, and realizations; geometric bins are used for
the randomized-reference-model and day-to-day experiments.



The lower support boundary is
$b_{\min}
=
\min_{\substack{s,v\\
0<d^G(s,v)<+\infty}}
d^G(s,v)$, where distances are measured in seconds. To determine an upper boundary covering the strongest perturbations, we construct a fixed reference realization in which all temporal edges are delayed independently by a value between $20$ and $30$ minutes. If $\mathcal G^{\mathrm{ref}}$ denotes this reference network, we set $b_{\max}
=
(1+\varepsilon)
\max_{s,v}
d^{\mathcal G^{\mathrm{ref}}}(s,v)$,
where $\varepsilon>0$ is a small numerical constant ensuring that the maximum distance lies inside the last bin. The reference realization is generated using a fixed random seed and is used only to determine the common support.
Our first family of experiments evaluates whether the dissimilarity responds consistently to the extent and magnitude of a controlled temporal perturbation.

Let $\mathcal{G}=(V,\mathcal{E})$ be a scheduled transportation network, with $m=|\mathcal{E}|$. For a perturbation level $p\in[0,1]$, we select uniformly at random, without replacement, $\lfloor pm\rfloor$ temporal edges. For each selected edge $e=(u,v,t,\delta)$, we independently sample an integer delay $\epsilon_e\sim
\operatorname{Unif}\{60a,\ldots,60b\}$, where $a$ and $b$ are given in minutes, and replace it by $e'=(u,v,t+\epsilon_e,\delta)$. Thus, the connection is shifted in time while its traversal duration remains unchanged.

This perturbation models an event-level delay rather than an increase in traversal duration. Although the duration of the selected connection remains unchanged, its later availability may increase waiting times, invalidate a transfer, or cause a different temporal path to become optimal. Delays are applied independently to the selected timetable connections and are not propagated to subsequent connections belonging to the same vehicle. 

We consider the perturbation levels $p\in\{0.2,0.4,0.6,0.8,1.0\}$, corresponding to delays applied to $20\%$, $40\%$, $60\%$, $80\%$, and $100\%$ of the temporal edges. Figures also include \(p=0\) as the unperturbed reference. To separate the effect of perturbation extent from that of perturbation magnitude, we use the three delay intervals $[0,5]$, $[5,10]$, $[20,30]$ minutes. For each delay range, perturbation level, and dataset, we generate 50 independent realizations per configuration. Reported dissimilarity values are the mean over these 50 realizations.

\begin{table}[t]
\centering
\begin{tabular}{lrr}
\toprule
City & $n=|V|$ & $m=|E|$ \\
\midrule
Belfast    & 1,917 & 122,693 \\
Grenoble   & 1,547 & 114,492 \\
Luxembourg & 1,367 & 186,752 \\
Turku      & 1,850 & 133,512 \\
\bottomrule
\end{tabular}
\caption{Summary statistics of the four public-transport networks.}
\label{tab:datasets}
\end{table}

Controlled delays preserve the endpoints and durations of all temporal edges and modify only the availability times of selected connections. To isolate the effect of other temporal characteristics, we also consider four randomized reference models. For a perturbation level $p$, every model first selects a uniformly random subset $\mathcal{E}_p
    =
    \{e_1,\ldots,e_k\}$, $k=\lfloor p|\mathcal{E}|\rfloor$, without replacement. Write $e_i=(u_i,v_i,t_i,\delta_i)$.
All models preserve the endpoints $(u_i,v_i)$ and leave edges outside
$\mathcal{E}_p$ unchanged.

\begin{description}

\item[\textsc{Duration Shuffling} (\textsc{DS}).] 
Let $\pi_{\delta}$ be a uniformly random permutation of $\{1,\ldots,k\}$. Each selected edge is replaced by $
    e_i'
    =
    (u_i,v_i,t_i,\delta_{\pi_{\delta}(i)})$. The departure time of every selected event remains fixed, while the
multiset of selected durations is preserved and reassigned among the
events.

\item[\textsc{Time Shuffling} (\textsc{TS}).]
Let $\pi_t$ be a uniformly random permutation of $\{1,$  $\ldots,k\}$. We set
$    e_i'
    =
    (u_i,v_i,t_{\pi_t(i)},\delta_i)
$.
Thus, the selected departure times are permuted while every edge retains its original duration. The arrival time is recomputed as
$t_{\pi_t(i)}+\delta_i$.

\item[\textsc{Uniform Random Times} (\textsc{URT}).]
For every selected edge, a new departure time is independently drawn from the observation interval: $
t_i' \sim $      

$
\operatorname{Unif}\bigl(\{t_{\min},\ldots,t_{\max}\}\bigr).
$, where $t_{\min}
    =
    \min_{e\in\mathcal{E}}\operatorname{dep}(e)
    $ and $
    t_{\max}
    =
    \max_{e\in\mathcal{E}}\operatorname{dep}(e)$.
The transformed edge is $e_i'
    =
    (u_i,v_i,t_i',\delta_i)$.
This model preserves durations but replaces the empirical timing pattern by uniformly distributed departure times.

\item[\textsc{Time--Duration Decoupling} (\textsc{TDD}).]
Let $\pi_t$ and $\pi_{\delta}$ be independent uniformly random permutations of $\{1,\ldots,k\}$. We define $
    e_i'
    =
    \bigl(
        u_i,v_i,
        t_{\pi_t(i)},
        $ $\delta_{\pi_{\delta}(i)}
    \bigr)$.
This model preserves separately the multisets of selected departure times and durations, but destroys their association with each other and with the original events.

\end{description}


For each reference model, we use the same perturbation levels $p$ as above and generate $50$ independent realizations per configuration.

\noindent \textbf{FAD/JS baselines.} To compare the proposed measure with an existing temporal-network dissimilarity, we evaluate the FAD/JS framework of Zhan et al.~\cite{ZhLiWaWaHoZh} on the same controlled-delay realizations. For this baseline, every timetable connection $(u,v,t,\delta)$ is projected to the directed instantaneous contact $(u,v,t)$; hence, traversal duration is not part of the baseline representation. For every source, the first-arrival distance (FAD) is the number of temporal contacts on an earliest-arrival path. The resulting FAD distributions and temporal node dispersion are compared using the Jensen--Shannon-based temporal dissimilarity defined by Zhan et al. We retain the direction of transportation contacts and explicitly align temporally unreachable destinations when comparing the two networks. We denote this baseline by $\operatorname{TD}_{\mathrm{FAD/JS}}$. For both FAD/JS baselines, we set \(\omega_1=\omega_2=0.5\), following Zhan et al.

The Jensen–Shannon divergence between two probability distributions is bounded by \(\log 2\). Accordingly, in the formulation of Zhan et al. [13], the between-network term \(\sqrt{J(\mu_{G_1},\mu_{G_2})/\log 2}\) lies in \([0,1]\). Their temporal node dispersion, \(\mathrm{TNND}(G)=J(H_1,\ldots,H_N)/\log(l_{\max}+1)\), is likewise normalized to \([0,1]\), and hence \(|\sqrt{\mathrm{TNND}(G_1)}-\sqrt{\mathrm{TNND}(G_2)}|\in[0,1]\). Since \(TD_{\mathrm{FAD/JS}}\) combines these terms with non-negative weights satisfying \(\omega_1+\omega_2=1\), it is itself bounded in \([0,1]\). Thus, like the normalized \(TD_{\mathrm{fast}}\), the baseline has a bounded unit range, so differences in relative sensitivity and saturation cannot be attributed merely to different numerical scales.

We additionally evaluate a \emph{duration-aware FAD/JS adaptation}. In this variant, a connection $(u,v,t,\delta)$ is feasible only after the source path has reached $u$ by time $t$, and traversing it reaches $v$ at time $t+\delta$. Traversal durations therefore affect path feasibility and earliest arrival. The resulting FAD value nevertheless remains the number of timetable connections on an earliest-arrival path. This second comparison controls for the omission of traversal durations; the two approaches nevertheless still differ in their path objective and in how the selected path is summarized. For both baselines, the delayed edges, sampled delay values, perturbation levels, delay ranges, and random seeds are identical to those used for $\operatorname{TD}_{\mathrm{fast}}$. We denote this adaptation by $\operatorname{TD}_{\mathrm{DA\!-\!FAD/JS}}$.

To quantify delay-regime separability, for each dataset and \(p\) we perform a one-way ANOVA on the 50 realization-level dissimilarities in the \(0\!-\!5\), \(5\!-\!10\), and \(20\!-\!30\) minute groups. We use the resulting \(F\)-statistic and \(\eta^2=SS_{\mathrm{between}}/SS_{\mathrm{total}}\) as descriptive separability measures, with larger values indicating stronger between-regime separation relative to within-regime variation.

In addition to comparisons between an observed network and a randomized or delayed version, we compare temporal segments of the same transportation system.
For each week-long dataset, we extract Friday, Saturday, and Sunday networks and compute all three pairwise dissimilarities using a common vertex universe and profile support. Stops absent from one day are treated as isolated vertices, so their unreachable destinations contribute to the unreachable state.




All repeated experiments used explicit random seeds for reproducibility, and the 1-Wasserstein distance was computed using SciPy’s standard implementation.


\subsection{Experimental Results}
\label{sec:experimental-results}
\begin{figure*}[t]
    \centering
    \includegraphics[width=\textwidth]
    {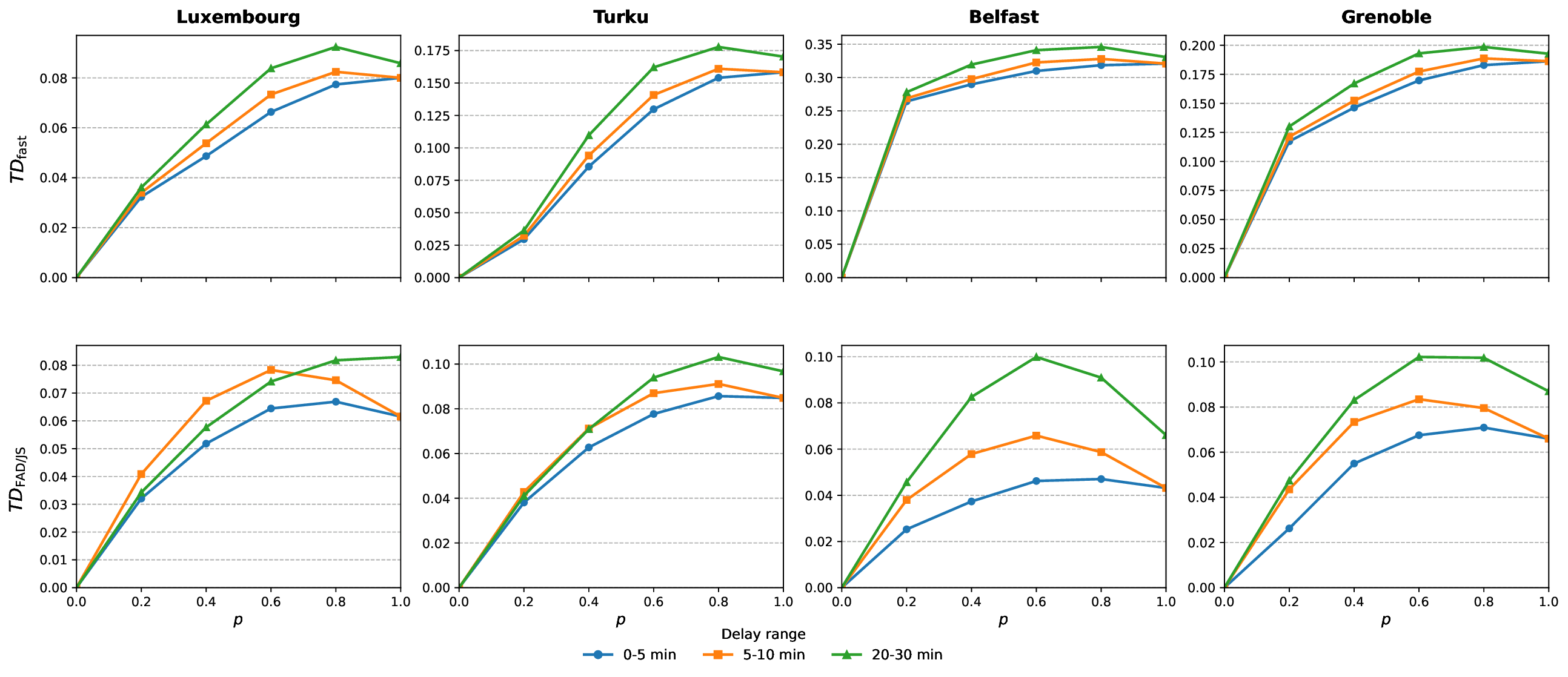}
    \caption{Controlled-delay comparison of $\operatorname{TD}_{\mathrm{fast}}$ and $\operatorname{TD}_{\mathrm{FAD/JS}}$ across the transportation networks, using the same perturbation realizations and delay ranges.}
    \label{fig:controlled-delays}
\end{figure*}

We report results for $\operatorname{TD}_{\mathrm{fast}}$ under controlled delays, duration-aware randomized reference models, and day-to-day variations.

\begin{figure*}[t]
    \centering
    \includegraphics[width=0.85\textwidth]
    {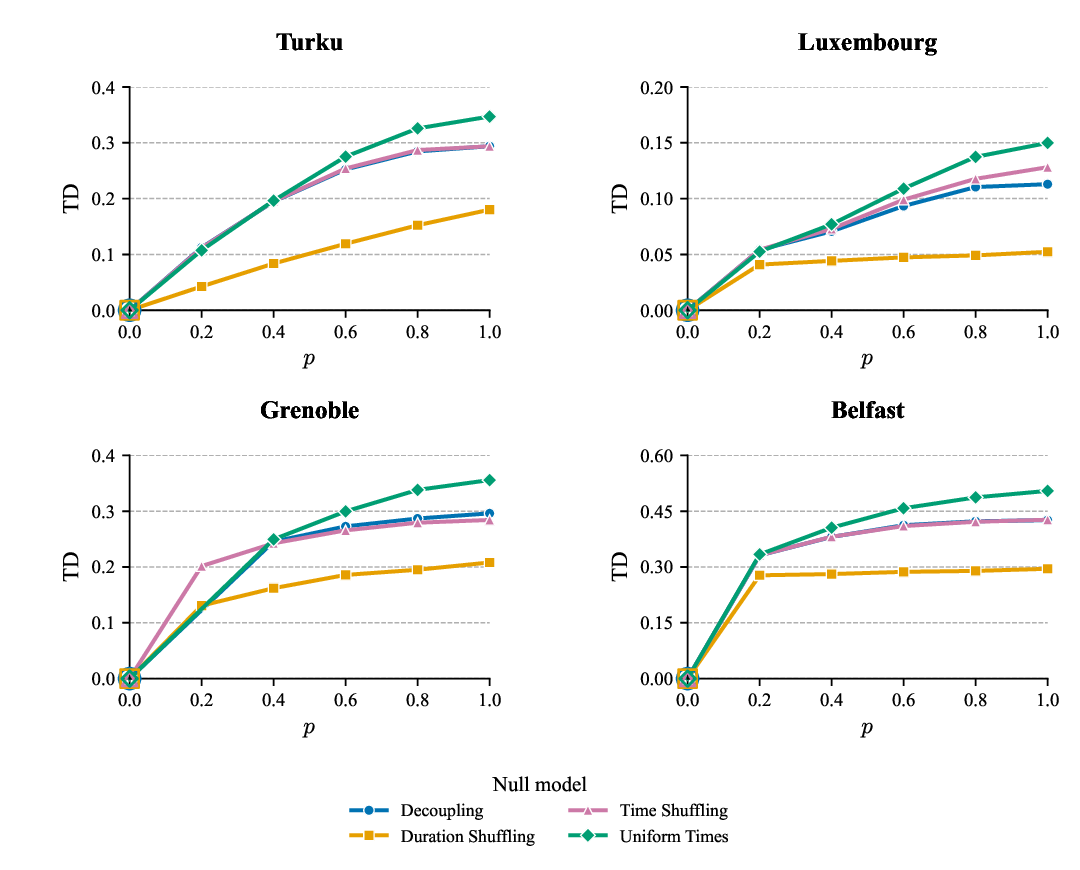}
    \caption{Sensitivity of $\operatorname{TD}_{\mathrm{fast}}$ to four randomized reference models across the transportation networks.}
    \label{fig:null-models}
\end{figure*}

\noindent \textbf{Controlled delays and comparison with FAD/JS.} Figure~\ref{fig:controlled-delays} compares $\operatorname{TD}_{\mathrm{fast}}$ with the $\operatorname{TD}_{\mathrm{FAD/JS}}$ as both the fraction of delayed edges and the delay magnitude increase. For every city and configuration, both measures are evaluated on the same perturbed timetable realizations.


Both measures detect the controlled perturbations: their dissimilarity generally increases as a larger fraction of the timetable is delayed, with a tendency to saturate for large $p$. This is important because it shows that the FAD/JS baseline is itself sensitive to substantial changes in the temporal organization of the network. The distinction appears primarily in the response to delay magnitude. $\operatorname{TD}_{\mathrm{fast}}$ maintains a comparatively clear separation between the $0$--$5$, $5$--$10$, and $20$--$30$ minute regimes across the datasets, with the largest delays generally producing the largest dissimilarity. In the FAD/JS baseline, the separation is less uniform: the ordering of the delay ranges is not preserved throughout all configurations, most visibly in Luxembourg, and several curves become close as the perturbation becomes widespread. 

This pattern is also supported quantitatively: \(TD_{\mathrm{fast}}\) yields larger \(F\)-statistics and \(\eta^2\) values than \(TD_{\mathrm{FAD/JS}}\) in 14 of the 20 dataset–\(p\) configurations, including all ten configurations in Luxembourg and Turku. Across all configurations, the median \(F\)-statistic is \(798.25\) for \(TD_{\mathrm{fast}}\) versus \(326.74\) for \(TD_{\mathrm{FAD/JS}}\), with corresponding median \(\eta^2\) values of \(0.915\) and \(0.816\).

This difference reflects the information retained by the two representations. FAD summarizes an earliest-arrival path by its number of contacts, so a larger delay may substantially change elapsed journey time without a proportional change in FAD. $\operatorname{TD}_{\mathrm{fast}}$, instead, directly represents elapsed fastest-path duration, including traversal and waiting time.



\noindent \textbf{Duration-Aware FAD/JS Comparison.}
The original FAD/JS baseline represents each temporal interaction as
an instantaneous contact and therefore does not receive traversal
duration as input. To control for this representational difference, we
also considered a duration-aware adaptation in which an edge
$(u,v,t,\delta)$ can be traversed only when $u$ has been reached by
time $t$, and the corresponding arrival time is $t+\delta$.
Earliest-arrival paths are therefore computed with the same traversal
durations as in the transportation network. The FAD assigned to a
destination, however, remains the number of timetable connections on
an earliest-arrival path.

\begin{figure*}[h]
    \centering
    \includegraphics[width=\textwidth]
    {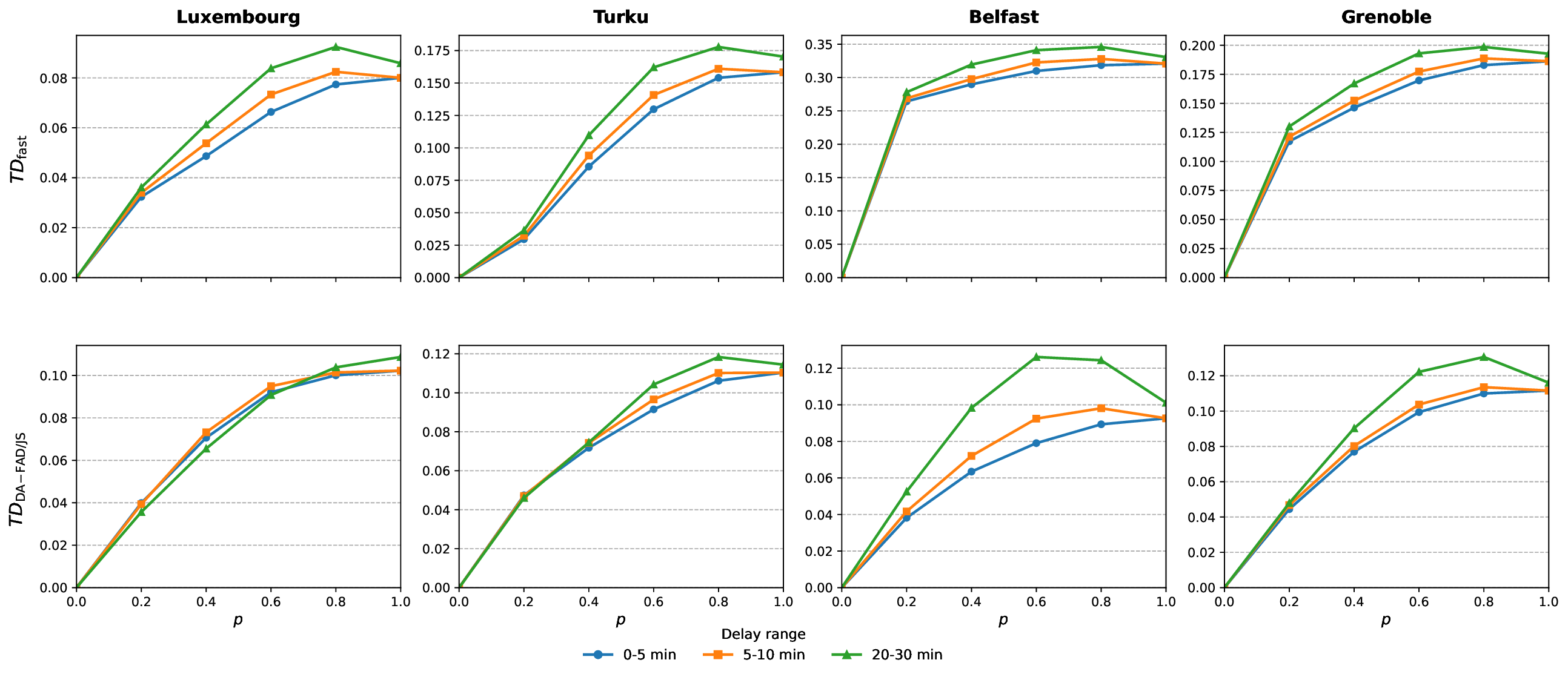}
    \caption{Controlled-delay comparison between
    $\operatorname{TD}_{\mathrm{fast}}$ and the duration-aware
    adaptation of the FAD/JS framework. Both methods use the same
    duration temporal networks and exactly the same perturbation
    realizations.}
    \label{fig:duration-aware-fad}
\end{figure*}

Figure~\ref{fig:duration-aware-fad} shows that incorporating traversal
durations into FAD improves the semantic alignment of the two
approaches, since both now use duration information when determining
temporal feasibility. Nevertheless, the FAD/JS response remains
comparatively compressed across delay magnitudes, especially in
Luxembourg and Turku. This is consistent with the fact that FAD still
summarizes an earliest-arrival path by its number of connections,
whereas $\operatorname{TD}_{\mathrm{fast}}$ directly represents
elapsed journey duration.

Quantitatively, $\operatorname{TD}_{\mathrm{fast}}$ yields larger
$F$-statistics and $\eta^2$ values than
$\operatorname{TD}_{\mathrm{DA-FAD/JS}}$ in 19 of the 20
dataset--$p$ configurations, with median $F$-statistics of $798.25$
versus $49.21$ and median $\eta^2$ values of $0.915$ versus $0.401$.

\noindent \textbf{Sensitivity to randomized reference models.} We next investigate which temporal ingredients contribute most strongly to the distance profiles by selectively destroying different forms of temporal organization. Figure~\ref{fig:null-models} compares the four duration-aware randomized reference models introduced in Section~\ref{sec:experimental-methodology}. The figure reports the results obtained with geometrically spaced bins.

A consistent ordering emerges across the four datasets. Duration Shuffling produces the smallest dissimilarity, whereas Uniform Random Times yields the largest values for moderate and large $p$; Time Shuffling and Time–Duration Decoupling generally lie between these extremes. Thus, perturbing the temporal placement of interactions affects fastest-distance profiles more strongly than reassigning traversal durations while preserving the empirical departure schedule.


Duration Shuffling also highlights why explicit traversal durations matter for network comparison. This transformation leaves every projected instantaneous contact $(u,v,t)$ unchanged and modifies only the traversal-duration assignment. Consequently, the original contact-based FAD/JS representation is identical before and after Duration Shuffling, even though the induced fastest journey times may change. The non-zero response of $\operatorname{TD}_{\mathrm{fast}}$ therefore captures temporal information that is absent from the instantaneous-contact projection by construction. The stable ordering across datasets indicates that the profiles distinguish
which temporal feature is perturbed, rather than merely the number of modified edges.



\begin{table}[t]
\centering
\small
\setlength{\tabcolsep}{3.5pt}
\renewcommand{\arraystretch}{0.88}

\begin{tabular}{l l c c}
\toprule
\textbf{City} & \textbf{Day} & \textbf{Saturday} & \textbf{Sunday} \\
\midrule

\multirow{2}{*}{Belfast} 
 & Friday   & 0.202 & 0.493 \\
 & Saturday & ---   & 0.372 \\
\cmidrule(lr){1-4}

\multirow{2}{*}{Grenoble} 
 & Friday   & 0.072 & 0.323 \\
 & Saturday & ---   & 0.291 \\
\cmidrule(lr){1-4}

\multirow{2}{*}{Luxembourg} 
 & Friday   & 0.045 & 0.238 \\
 & Saturday & ---   & 0.206 \\
\cmidrule(lr){1-4}

\multirow{2}{*}{Turku} 
 & Friday   & 0.120 & 0.179 \\
 & Saturday & ---   & 0.088 \\

\bottomrule
\end{tabular}

\caption{Pairwise fastest-path dissimilarities between Friday, Saturday, and Sunday timetable networks.}
\label{weekday-comparison}
\end{table}

\noindent\textbf{Effect of temporal-path semantics.}
To examine the effect of the temporal-distance notion, we also evaluate the randomized reference-model experiments using minimum-hop temporal distance. The randomized networks are unchanged; only the temporal-distance notion used to construct the source profiles differs.

For fastest temporal distance, travel edges are weighted by their
traversal durations and waiting edges by the corresponding waiting
times. For minimum-hop distance, travel edges have weight one and
waiting edges have weight zero. Thus, the former measures elapsed
journey time, whereas the latter counts the number of scheduled
connections used.

This comparison shows whether changes in temporal performance are also reflected in route complexity. A delay may, for example, change the fastest travel time while leaving the minimum number of required
connections unchanged.

\begin{figure*}[h]
    \centering
    \includegraphics[width=0.7\textwidth]
    {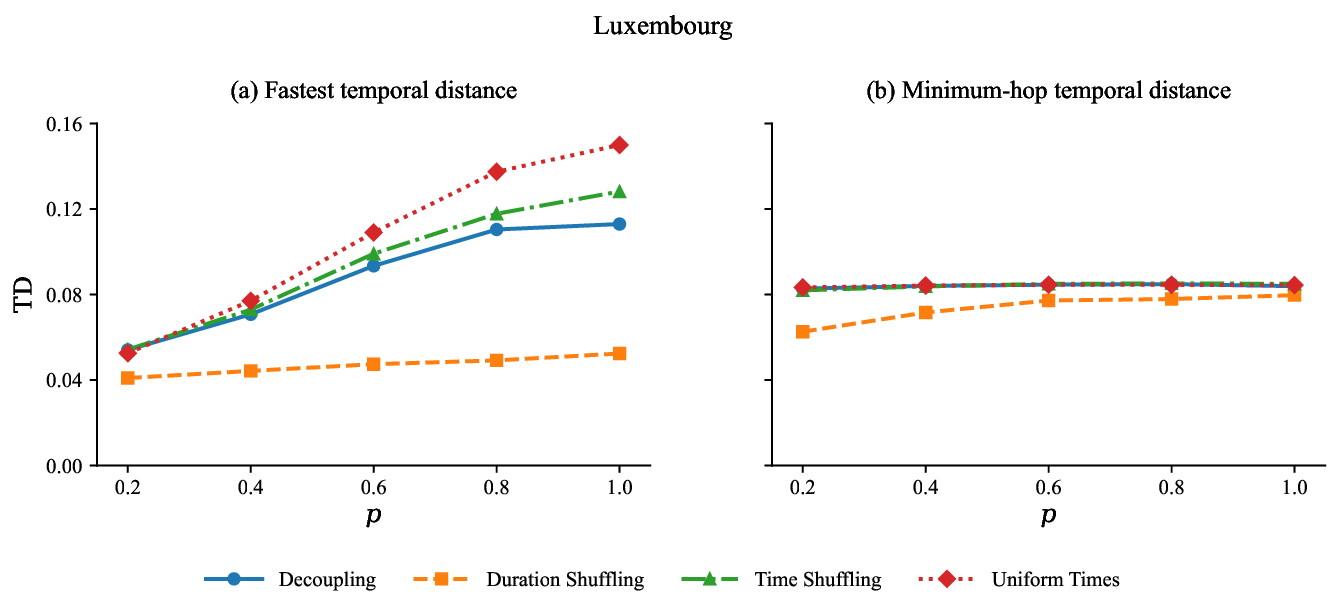}
    \caption{Comparison of the response of the four randomized reference models under fastest and minimum-hop temporal distances.}
    \label{fig:comparison}
\end{figure*}

\noindent \textbf{Natural day-to-day variability.} Table~\ref{weekday-comparison} shows that Friday--Sunday is the most dissimilar pair in all four datasets, while the closest pair varies by city. Thus, the measure also captures naturally occurring timetable variation.


\section{Conclusion} 
We introduced a duration-aware framework for comparing temporal networks with non-zero traversal times using source-specific temporal-distance profiles. Experiments on four public-transport networks show that $\operatorname{TD}_{\mathrm{fast}}$ responds consistently to both the extent and magnitude of temporal perturbations. Comparisons with FAD/JS baselines indicate that retaining elapsed journey-time information provides a more consistent distinction between delay magnitudes, beyond the mere availability of traversal durations. Overall, the results support duration-aware temporal distances as a useful representation for comparing temporal systems in which traversal and waiting times affect network performance.


\begin{thebibliography}{10}

\bibitem{BaBo}
Bagrow, J. P., \& Bollt, E. M. (2019). An information-theoretic, all-scales approach to comparing networks. Applied Network Science, 4(1), 1-15.

\bibitem{Ba}
Bast, H., Delling, D., Goldberg, A., Müller-Hannemann, M., Pajor, T., Sanders, P., Wagner, D., \& Werneck, R. F. (2016). Route planning in transportation networks. In Algorithm engineering: Selected results and surveys (pp. 19-80). Cham: Springer International Publishing.


\bibitem{GaGeKaKiTaVaVe}
Gauvin, L., Génois, M., Karsai, M., Kivela, M., Takaguchi, T., Valdano, E., \& Vestergaard, C. L. (2022). Randomized reference models for temporal networks. Siam Review, 64(4), 763-830.

\bibitem{FrJaNiRe}
Froese, V., Jain, B., Niedermeier, R., \& Renken, M. (2020). Comparing temporal graphs using dynamic time warping. Social Network Analysis and Mining, 10(1), 50.

\bibitem{HoSa}
Holme, P., \& Saramäki, J. (Eds.). (2019). Temporal network theory (Vol. 2). New York: Springer.

\bibitem{OrRo}
Orda, A., \& Rom, R. (1990). Shortest-path and minimum-delay algorithms in networks with time-dependent edge-length. Journal of the ACM (JACM), 37(3), 607-625.

\bibitem{KeKlKu}
Kempe, D., Kleinberg, J., \& Kumar, A. (2000). Connectivity and inference problems for temporal networks. In Proceedings of the thirty-second annual ACM symposium on Theory of computing (pp. 504-513).

\bibitem{KuWeDaMlSa}
Kujala, R., Weckström, C., Darst, R. K., Mladenović, M. N., \& Saramäki, J. (2018). A collection of public transport network data sets for 25 cities. Scientific data, 5(1), 180089.

\bibitem{GaMa}
Peyré, G., \& Cuturi, M. (2019). Computational optimal transport with applications to data sciences. Foundations and trends in machine learning, 11(5-6), 355-607.

\bibitem{PyScWaZa}
Pyrga, E., Schulz, F., Wagner, D., \& Zaroliagis, C. (2008). Efficient models for timetable information in public transportation systems. Journal of Experimental Algorithmics (JEA), 12, 1-39.

\bibitem{ScCaDiOaNaRa}
Schieber, T. A., Carpi, L., Díaz-Guilera, A., Pardalos, P. M., Masoller, C., \& Ravetti, M. G. (2017). Quantification of network structural dissimilarities. Nature communications, 8(1), 13928.

\bibitem{WuChHuKeLuXu}
Wu, H., Cheng, J., Huang, S., Ke, Y., Lu, Y., \& Xu, Y. (2014). Path problems in temporal graphs. Proceedings of the VLDB Endowment, 7(9), 721-732.

\bibitem{ZhLiWaWaHoZh}
Zhan, X. X., Liu, C., Wang, Z., Wang, H., Holme, P., \& Zhang, Z. K. (2025). Measuring and utilizing temporal network dissimilarity. Communications Physics, 8(1), 40.

\end{thebibliography}
\end{document}